\documentclass[11pt]{article}

\usepackage[margin=1.3in]{geometry}
\usepackage{amsmath, amssymb, amsthm}
\usepackage{enumerate}

\newtheorem{proposition}{Proposition}

\newtheorem{lemma}{Lemma}
\newtheorem{theorem}{Theorem}

\newtheorem{definition}{Definition}

\def\A{\mathcal{A}}
\def\Y{\mathcal{Y}}

\def\B{\mathbf{B}}
\def\mb{\mathbin}

\begin{document}

\title{The Empirical Implications of Privacy-Aware Choice\thanks{We thank Jeff Ely, and especially Katrina Ligett, for useful comments. This work was supported by NSF grants 1101470, 1216006, and CNS-1254169, as well as BFS grant 2012348 and the Charles Lee Powell Foundation.}}
\author{Rachel Cummings\thanks{Electrical Engineering and Computer Science, Northwestern University. {\tt rachelc@u.northwestern.edu}} \and Federico Echenique\thanks{Humanities and Social Sciences, California Institute of Technology. {\tt fede@hss.caltech.edu}} \and Adam Wierman\thanks{Computing and Mathematical Sciences, California Institute of Technology. {\tt adamw@caltech.edu}}}

\maketitle

\begin{abstract}

This paper initiates the study of the testable implications of choice
data in settings where agents have privacy preferences.  We adapt the
standard conceptualization of consumer choice theory to a situation
where the consumer is aware of, and has preferences over, the
information revealed by her choices.  The main message of the paper is
that little can be inferred about consumers' preferences once we
introduce the possibility that the consumer has concerns about
privacy.  This holds even when consumers' privacy preferences are
assumed to be monotonic and separable.  This motivates the
consideration of  stronger assumptions and, to that end, we introduce
an additive model for privacy preferences that does have testable
implications.

\end{abstract}

\section{Introduction}
The purpose of this paper is to study what an observer can learn about a
consumer's preferences and behavior when the consumer has concerns for her
privacy and knows that she is being observed. The basic message of our
results is that very little can be learned without strong assumptions
on the form of the consumer's privacy preferences.

To motivate the problem under study, consider the following story.
Alice makes choices on the internet. She chooses which websites to
visit, what books to buy, which hotels to reserve, and which
newspapers to read. She knows, however, that she is being watched. An
external agent, ``Big Brother'' (BB), monitors her choices. BB could
be a private firm like Google, or a government agency like the NSA. As
a result of being watched, Alice is concerned for her privacy; and
this concern affects her behavior.

Alice has definitive preferences over the things she chooses
among. For example, given three political blogs $a$, $b$, and $c$, she may
prefer to follow $a$. But, BB will observe such a choice, and infer that
she prefers $a$ over $b$ and $c$. This is uncomfortable to Alice,
because her preferences are shaped by her political views, and she does
not like BB to know her views or her preferences. As a result, she
may be reluctant to choose $a$. She may choose $b$ instead because
she is more comfortable with BB believing that she ranks $b$ over $a$
and $c$.\footnote{Like Alice, 85\% of adult internet users have take
  steps to avoid surveillance by other people or organizations, see
  \cite{pewsurvey}.}

Now, the question becomes, given observations of Alice's behavior, what can we learn about her preferences? We might conjecture that her
behavior must satisfy some kind of rationality axiom, or that one
could back out, or reverse-engineer, her preferences from her
behavior. After all, Alice is a fully rational
consumer (agent), meaning  that she maximizes a utility function (or a transitive
preference relation). She has a well-defined preference over the objects of
choice, meaning that if she could fix what BB learns about her---if
what BB learns about her were independent from her choices---then she
would choose her favorite object. Further, Alice's
preferences over privacy likely satisfy particular structural
properties.  For example, she has well-defined preferences over the
objects of choice, and she cares about the preference revealed by
her choices: she always prefers revealing less to revealing
more. In economics, preferences of this form are called separable and
monotonic; and such preferences normally place strong
restrictions on agents' behavior.

However, contrary to the above discussion, the results in this paper
prove that nothing can be inferred about Alice's preferences once we
introduce the possibility that she has concerns about privacy. No
matter what her behavior, it is compatible with some concerns over
privacy, i.e., she always has an ``alibi'' that can explain her
choices as a consequence of privacy concerns. The strongest version of
this result is that {\em all possible  behaviors on the part of Alice
  are compatible with all possible   preferences that Alice may have
  over objects}:  postulate some arbitrary  behavior for Alice, and
some arbitrary preference over objects, and the two will always be
compatible.

So BB's objective is
hopeless. He can never learn anything about Alice's true preferences
over political blogs, or over any other objects of choice. If BB tries
to estimate preferences from some given choices by Alice, he finds
that all preferences could be used to  explain her choices. He cannot
narrow down the set of preferences Alice might have, no matter what
the observed behavior. The result continues to hold if BB
adversarially sets up scenarios for Alice to choose from.  That is, even if BB
offers Alice menus of choices so as to maximize what he can learn from
her behavior, the result is still that nothing can be learned.

The results in this paper have a variety of implications.

First, they motivate the use of specific parametric models of preferences over
privacy. Our main result makes strong qualitative assumptions about
preferences (separability, monotonicity). Given that such assumptions
lack empirical bite, one should arguably turn to stronger
assumptions yet. The paper proposes an additive utility function that
depends on the chosen object and on what is revealed by the consumer's
choices. If Alice chooses $x$ then she obtains a utility $u(x)$ and a
``penalty'' $v(x,y)$ for not choosing $y$, for all non-chosen $y$,
as she reveals to BB that she ranks $x$ over $y$.
This additive model does have restrictions for the consumer's
behavior, and could be estimated given data on Alice's choices. The
model is methodologically close to models used in economics to explain
individual choices, and could be econometrically estimated using
standard techniques. The paper discusses a test for the
additive model based on a linear program.

Second, while the paper's main motivation is consumers' behavior on the
internet, the results have implications for  issues commonly
discussed in behavioral economics. Some behavioral ``anomalies'' could
be the consequence of the presence of an outside observer. For example
(elaborating on a laboratory experiment by \cite{simonson1992choice2}),
consider a consumer who is going to buy a technical gadget, such as
a phone or a camera. The consumer might prefer a simple camera over a
more complex one which they might not know how to operate; but when
presented with a menu that has a simple,
an intermediate and an advanced camera, they might choose the
intermediate one because they do not want to reveal to the world that
they do not know how to use a complex camera. Of course, the results
show that this line of reasoning may not be very useful, as anything can
be explained in this fashion. The results suggest, however, that a
stronger parametric model may be useful to explain various behavioral phenomena.

Third, the results explain why BB may want to be hide the
fact that consumer behavior is being observed. The NSA or Google
seem to dislike openly discussing that they are monitoring consumers' online
behavior. One could explain such a desire to hide by political issues,
or because the observers wish to maintain a certain public image, but
here we point to another reason. The observations simply become
ineffective when the consumer is aware that she is being observed.

\section{Modeling privacy preferences}
\label{s.modeling}

The goal of this paper is to study the testable implications of choice
data in a context where agents have privacy preferences.  To this end,
we adapt the standard conceptualization of consumer choice theory in
economics (see e.g. the textbook treatments in \cite{mas1995} or
\cite{rubinstein2012lecture}) to a
situation where the consumer is aware of, and has preferences over,
the information revealed by her choices.

\subsection{The setting}
We focus on a situation where there is an outside observer (he), such
as Google or the NSA, that is gathering data about the choices of a
consumer (she) by observing her choices.  We assume that the consumer
is presented with a set of alternatives $A$ and then makes a choice
$c(A)$, which the outside observer sees.  The observer then infers
from this choice that $c(A)$ is preferred to all other alternatives in
$A$.

The above parallels the classical revealed preference theory
framework; however our model differs when it comes to the the behavior
of the consumer, which we model as ``privacy-aware''.  We assume that
the consumer is aware of the existence of an outside observer, and so
she may care about what her choices reveal about her. Specifically,
her choices are motivated by two considerations. On the one hand, she
cares about the actual chosen alternative. On the other hand, she
cares about what those choices reveal about her preferences over
alternatives, i.e., her revealed preferences. We capture this by
assuming that the consumer has preferences over pairs $(x,B)$, where
$x$ is the chosen object and $B$ is the information revealed about the
consumer's preferences.

An important point about the setting  is that the
inferences made by the observer do not recognize that the consumer is
privacy aware. This assumption about the  observer being naive is
literally imposed on the behavior of the observer, but {\em it is really an
assumption about how the agent thinks that the observer makes
inferences.} The agent thinks that the observer naively uses revealed
preference theory to make inferences about her preferences. The
observer, however, could be as sophisticated as any reader of this
paper in how they learn about the agent's preferences. The upshot of
our results is that such a sophisticated observer could not learn
anything about the agent's behavior.

It is natural to go one step further and  ask ``What if the agent
knows that the observer knows that the agent is privacy-aware?''
Or, ``what if the agent knows that the observer knows that the agent
knows that the observer knows that the agent is privacy-aware?''
The problem naturally lends itself to a discussion of the role of
higher order beliefs. We formalize exactly this form of a cognitive
hierarchy in Section~\ref{s.higherorder}, and we discuss how our
results generalize.

\subsection{Preliminaries}

Before introducing our model formally, there are a few preliminaries that
are important to discuss.  Let $\B(X) = 2^{X\times X}$
denote the set of all binary preference relations on a set $X$ and
recall that a binary relation $\succeq$ is a weak order if it is
complete (total) and transitive. We say that $x \succ y$ when $x
\succeq y$ and it is not the case that $y \succeq x$.  Finally, a
linear order is a weak order such that if $x\neq y$ then $x\succ y$ or
$y\succ x$.

We shall often interpret binary relations as graphs. For $B\in
\B(X)$, define a graph by  letting the vertex set
of the graph be equal to $X$ and the edge set be $B$. So,
 for each element $(x,y) \in B$, we
have a directed edge in the graph from $x$ to $y$.  We say that a
 binary relation $B$ is acyclic if there does not exist a directed
 path that both originates and ends at $x$, for any $x \in X$.
The following simple result, often called Spilrajn's Lemma, is useful.

\begin{lemma} \label{lem:spil}
If $B\subseteq \B(X)$ is acyclic, then there is a linear
  order $\succeq$ such that $B\subseteq \succeq$.
\end{lemma}

\subsection{The model}

Given the setting described above, our goal is to characterize the
testable implications of choice data, and to understand how the
testable implications change when consumers are privacy-aware as
opposed to privacy-oblivious.  To formalize this we denote a
\textit{choice problem} by a tuple $(X,\A,c)$ where $X$ is a finite set of
alternatives,  $\A$ a collection of nonempty
subsets of $X$, and  $c:\A \rightarrow X$ such that $c(A)\in A$ for
all $A\in \A$.

In choice problem $(X,\A,c)$, the consumer makes choices for each
$A\in\A$ according to the function $c$. Further, given $A\in \A$ and
$x=c(A)$, the observer infers that the consumer
prefers $x$ to any other alternative available in $A$. That is, he
infers that the binary comparisons $(x,y)$ $\forall$ $y\in
A\setminus\{x\}$ are part of the consumer's preferences over
$X$. Such inferences lie at the heart of revealed preference theory (see
e.g.\ \cite{varia82} or \cite{varian2006revealed}).

A \textit{privacy preference} is a linear order  $\succeq$ over  $X\times
2^{X\times X}$. A privacy preference ranks objects of the form
$(x,B)$, where $x\in X$ and $B\in \B(X)$.
If a consumer's choices are guided by a privacy
preference, then she cares about two things: she cares about the choice
made (i.e. $x$) and about what her choices reveal about her preference
(hence $B$).

Given the notions of a choice problem and privacy preferences defined
above, we can now formally define the notion of rationalizability that
we consider in this paper.

\begin{definition} \label{def:rational}
A choice problem $(X,\A,c)$ is \textbf{rationalizable (via privacy preferences)} if there is a privacy preference $\succeq$ such that if $x=c(A)$ and $y\in A\setminus\{x\}$ then
 \[
(x,\{(x,z):z\in A\setminus\{x\} \})\succ
(y,\{(y,z):z\in A\setminus\{y\} \}),
\]  for all $A\in\A$.
In this case, we say that $\succeq$ \textbf{rationalizes} $(X,\A,c)$.
\end{definition}

Thus, a choice problem is rationalizable when there exists a privacy
preference that ``explains'' the data, i.e., when there exists a
privacy preference for which the observed choices are maximal.

\section{The rationalizability of privacy-aware choice}
\label{s.mainresults}

In this section, we present our main results, which characterize when
choice data from privacy-aware consumers is rationalizable.  Our
results focus on the testable implications of structural assumptions
about the form of the privacy preferences of the consumer. While a consumer's preferences may, in general, be a complex combination of preferences over the choices and revealed preferences, there are some natural properties that one may expect to hold in many situations.  In
particular, we focus on three increasingly strong structural
assumptions in the following three subsections: monotonicity,
separability, and additivity.

\subsection{Monotone privacy preferences}

A  natural assumption on privacy preferences
is ``monotonicity'', i.e., the idea that revealing less information is
always better.  Monotonicity of privacy preferences is a common
assumption in the privacy literature, e.g., see \cite{Xiao13} and \cite{NOS12}, but of
course one can imagine situations where it may not hold,  e.g., see \cite{CCKMV13} and \cite{NVX14}.

In our context, we formalize monotone privacy preferences as follows.
\begin{definition}
A binary relation $\succeq$ over $X\times 2^{X\times X}$ is a \textbf{monotone privacy preference} when
\begin{enumerate}[(i)]
\item $\succeq$ is a linear order, and
\item $B\subsetneq B'$ implies that $(x,B)\succ (x,B')$.
\end{enumerate}
\end{definition}

This definition formalizes the idea that revealing less information is
better.  In particular, if $B\subsetneq B'$, then fewer comparisons
are being made in $B$ than in $B'$, so $(x,B)$ reveals less
information to the observer than $(x,B')$.

Given the above definition, the question we address is ``what are the
empirical implications of monotone privacy preferences?''  That is,
``Is monotonicity refutable via choice data?'' The following
proposition highlights that monotonicity is \emph{not} refutable,
so any choice data has a monotone privacy preference that explains
it.

\begin{proposition} \label{prop:anyc}
Any choice problem is rationalizable via monotone privacy preferences.
\end{proposition}

\begin{proof}
We shall use the following notation:
\[A_x =  \{(x,y) : y\in A\setminus \{x\}\}.\]
Define a binary relation $E$ on $X\times \B(X)$ as follows: $(x,B) \mb E (x',B')$ if either $x=x'$ and $B\subsetneq B'$, or $x\neq x'$ and there is $A\in \A$ with $x=c(A)$, $x'\in A$ and $B=A_X$ while $B'=A_{x'}$.  It will be useful for our proof to think of $E$ as the edges of a directed graph $G=(V,E)$, where $V = X\times \B(X)$.  The edges where $x = x'$ result from the monotonicity requirement, and the edges where $x \neq x'$ result from the requirement that observed choices be rationalized.  For shorthand, we will call these edges ``monotone'' and ``rationalizing,'' respectively.  It should be clear that any linear order that extends $B$ (i.e any linear order $\succeq$ with $E\subseteq \succeq$) is a monotone privacy preference that rationalizes $(X,\A,c)$. By Lemma~\ref{lem:spil}, we are done if we show that $E$ is acyclic.  To prove that $E$ is acyclic, it is equivalent to show that the graph is acyclic.

By the definition of $E$, for any pair $(x,B) \mb E(x',B')$, the cardinality of $B$ must be at most that of $B'$, and if $x=x'$ then the cardinality must be strictly smaller due to monotonicity.  Hence, there can be no cycles containing monotone edges.

Thus any cycle must contain only rationalizing edges $(x,B) \mb E (x',B')$ with $x\neq x'$. Each such edge arises from some $A\in \A$ for which $B=A_X$ while $B'=A_{x'}$, and for each such $A$ there is a unique $x\in A$ with $x=c(A)$.  If the graph were to contain two consecutive rationalizing edges, it would contradict uniqueness of choice.  Therefore there cannot be any cycles in $E$.
\end{proof}

Proposition \ref{prop:anyc} provides a contrast to the context of
classical revealed preference theory, when consumers are
privacy-oblivious. In particular, in the classical setting, choice
behavior that violates the strong axiom of revealed preferences
(SARP) is not rationalizable, and thus refutes the consumer
choice model.  However, when privacy-aware consumers are considered,
such a refutation of monotonic preferences is impossible.
Interestingly, this means that while one may believe that preferences
are non-monotonic, the form of data considered in this paper does not
have the power to refute monotonicity.\footnote{This phenomenon is common in
the consumer choice formulation of the revealed preference problem,
but it comes about for completely different reasons.}

Note that the question addressed by Proposition \ref{prop:anyc} is
only whether the consumer's choice behavior is consistent with
rational behavior, and is not about whether the consumer's underlying
preferences over outcomes in $X$ can be learned.  In particular, these
underlying preferences may not even be well defined for the general
model considered to this point.  We address this issue in the next
section after imposing more structure on the privacy preferences.

\subsection{Separable privacy preferences}\label{s.separable}
That all choice behavior is rationalizable via monotone
privacy preferences can be attributed to the flexibility provided by
such preferences. Here we turn to a significant restriction on the
preferences one might use in rationalizing the consumer's behavior.

It is natural to postulate
that the consumer would have some underlying, or intrinsic,
preferences over possible options when her choices are not observed.
Indeed, the observer is presumably trying to learn the agent's
preferences {\em over objects}. Such preferences should be well defined:
if outcome $x$
is preferred to outcome $y$ when both are paired with the same privacy
set $B$, then it is natural that $x$ will always be preferred to $y$
when both are paired with the same privacy set $B^{\prime}$, for all
possible $B^{\prime}$.  This property induces underlying preferences
over items in $X$, as well as the agent's privacy-aware
preferences.

We formalize the notion of separable privacy preferences as follows.
\begin{definition}
A binary relation $\succeq$ over $X\times 2^{X\times X}$ is a
\textbf{separable privacy preference} if it is a  monotone privacy
preference
and additionally satisfies that for all $x,y\in X$ and $B\in \B(X)$,
\[
(x, B) \succeq (y,B) \Longrightarrow (x, B^{\prime}) \succeq (y,B^{\prime}) \; \forall B^{\prime}\in\B(X)\]
\end{definition}
That is, whenever $(x,B)$ is preferred to $(y,B)$ for some preference
set $B$, then also $(x, B^{\prime})$ is preferred to $(y,B^{\prime})$
for all other  sets $B^{\prime}$.

Separable privacy preferences have an associated preference relation
over $X$. If $\succeq$ is a separable privacy preference, then define
$\succeq|_X$ as $x \succeq|_X y$ if and only if $(x,B)
\succsim (y,B)$ for all $B\in \B(X)$. Note that  $\succeq|_X$ is a
linear order over $X$. We can interpret $\succeq|_X$ as the projection
of $\succeq$ onto $X$.

There are two questions we seek to answer:  ``What are the empirical
implications of separability?''  and ``When can an observer learn the
underlying choice preferences of the consumer?''  The following
proposition addresses both of these questions.  Note that Proposition
~\ref{prop:nonid} follows from a more general result,
Theorem~\ref{thm:main}, which is presented in Section
\ref{s.higherorder}.

\begin{proposition} \label{prop:nonid}
Let $(X,\A,c)$ be a choice problem, and let $\succeq$ be any linear
order over $X$. Then there is a separable privacy preference
$\succeq^*$ that rationalizes $(X,\A,c)$ such that the projection
of $\succeq^*$ onto $X$ is well defined and coincides with $\succeq$,
i.e., $\succeq^*|_X = \succeq$.
\end{proposition}

Think of $\succeq$  as a conjecture that the observer has about the
agent. Proposition~\ref{prop:nonid} implies that {\em no matter the
  nature of such a  conjecture, and no matter what choice behavior is
  observed, the two are compatible.}

This proposition carries considerably more weight than Proposition
\ref{prop:anyc}.  Separability imposes much more structure than
monotonicity alone and, further, Proposition \ref{prop:nonid} says
much more than simply that separability has no testable implications,
or that it is not refutable via choice data. Proposition
\ref{prop:nonid} highlights that the task of the observer is hopeless
in this case -- regardless of the choice data, there are preferences
over revealed information that allow all possible choice
observations to be explained.

That is, the choice data does not
allow the observer to narrow his hypothesis about the consumer
preferences at all.  This is because the consumer always has an alibi
available (in the form of preferences over revealed information) which
can allow her to make the observed data consistent with any preference
ordering over choices.

In some sense, our result is consistent with the idea that secrecy is
crucial for observers such as the NSA and Google.  If the consumer is
not aware of the fact that she is being observed then the observer can
learn a considerable amount from choice data, while if the consumer is
aware that she is being observed then the choice data has little power
(unless more structure is assumed than separability).

\subsection{Additive privacy preferences}\label{s.additive}

So far, we have seen that monotonicity and separability do not provide
enough structure to allow choice data to have testable implications or
to allow the observer to learn \emph{anything} about consumer
preferences over choices.  This implies that further structure must be
imposed for choice data to have empirical power.  To that end, we now
give an example of a model for privacy preferences where choice data
does have testable implications.  The model we consider builds on the
notion of separable privacy preferences and additionally imposes
additivity.

\begin{definition} \label{def:addpref}
A binary relation $\succeq$ over $X\times 2^{X\times X}$ is an \textbf{additive privacy preference} if there are functions $u:X\rightarrow
\mathbb{R}^+$ and $v:X\times X\rightarrow \mathbb{R}^+$ such that $(x,B)\succ (x',B')$ iff
\[
 u(x) - \sum_{(z,z')\in B} v(z,z') >
u(x') - \sum_{(z,z')\in B'} v(z,z').\]
\end{definition}

While monotonicity and separability are general structural properties of privacy preferences, the definition of additivity is much more concrete.  It specifies a particular functional form, albeit a simple and natural one.  In this definition, the consumer experiences utility $u(x)$ from the choice made and disutility $v(x,y)$ from the privacy loss of revealing that $x \succ y$ for every pair $(x,y) \in X \times X$.  Note that this form is an additive extension of the classical consumer choice model, which would include only $u$ and not $v$.

Moreover, this definition also satisfies both monotonicity and
separability, making it a strictly stronger restriction.  Monotonicity
is satisfied because the agent always experiences a \emph{loss} from
each preference inferred by the observer.  Namely, the range of $v$ is
restricted to non-negative reals, so for a fixed choice element, the
agent will always prefer fewer inferences to be made about her
preferences.\footnote{Monotonicity restricts to the case where people
  want to keep their preferences private. It may be interesting to
  explore in future work, the case where people are happy to reveal
  their information, e.g., conspicuous consumption.  Under additive preferences, this would correspond to allowing the range of $v$ to be all of $\mathbb{R}$.}
Separability is satisfied because utilities $u$ determine the linear
ordering over $X$, so for a fixed set of inferences made by the
observer, privacy preferences will correspond to the preferences
determined by $u$.

Of course there are a number of variations of this form that could also make sense, e.g., if the disutility from a revealed preference $(x,y)$ was only counted once instead of (possibly) multiple times due to multiple revelations in the choice data. This would correspond to a consumer maximizing a ``global'' privacy loss rather than optimizing online for each menu. However, this modeling choice requires the agent to know ex ante the set $\A$ of menus from which she will choose, and additional assumptions about the order in which the she faces these menus.  For our analysis we restrict to additive preferences as defined above.

Rationalizability of additive privacy preferences corresponds to the existence of functions $u$ and $v$, such that the observed choice behavior maximizes the consumer's utility under these functions. Here, it turns out the imposed structure on privacy preferences is enough to allow the model to have testable implications, as shown in the following proposition. 

\begin{proposition} \label{prop:addtest}
There exists a choice problem $(X,\A,c)$ that is not rationalizable with additive privacy preferences.
\end{proposition}

Proposition \ref{prop:addtest} highlights that, while monotonicity and separability cannot be refuted with choice data, additivity can be refuted.  To show this, we construct a simple example of choice data that cannot be explained with any functions $u$ and $v$.

\begin{proof}[Proof of Proposition \ref{prop:addtest}]
To construct an example that is not rationalizable via additive privacy preferences, we begin by defining the set of alternatives as $X=\{x,y,z,w \}$ and the choice data as follows.  It includes six observations:
$z=c(\{x,z \})$,
$x=c(\{x,y,z\})$,
$w=c(\{w,z \})$,
$z=c(\{w,y,z\})$,
$x=c(\{x,w \})$,
$w=c(\{x,y,w\})$.

To see that this choice data is not rationalizable suppose, towards a contradiction, that the pair $(u,v)$ rationalizes $c$. Then  $z=c(\{x,z \})$ implies that \[
u(z) - v(z,x) > u(x) - v(x,z),
\] while $x=c(\{x,y,z\})$
implies that  \[
u(z) - v(z,x) - v(z,y) < u(x) - v(x,z) - v(x,y).
\] Therefore $v(z,y) > v(x,y)$.

Similarly, we can argue that $w=c(\{w,z \})$ and
$z=c(\{w,y,z\})$ together imply that $v(w,y) > v(z,y)$, and $x=c(\{x,w \})$
and $w=c(\{x,y,w\})$ together imply that $v(x,y) > v(w,y)$. This gives us a contradiction and so proves that the choice data is not rationalizable.
\end{proof}

Given that the structure imposed by additive privacy preferences is
testable, the next task is to characterize data sets that are
consistent with (or refute) the additive privacy preference model.
The example given in the proof of Proposition \ref{prop:addtest}
already suggests an important feature of choice data that must hold
for it to be rationalizable.

Given a choice problem $(X,\A,c)$  and an element $y\in X$, define the
 binary relation $R^y$ by $x\mathbin{R^y} z$ if there is $A\in
\A$ with $z=c(A)$ and $x= c(A\cup\{y\})$. Our next result gives a test
for additively rational preferences. It says that, if there are cycles
in the binary relation $R^y$, then the choice data cannot be
rationalized by additive privacy preferences.

\begin{proposition} \label{prop:neccond}
A choice problem can be rationalized by additive privacy preferences only if $R^y$ is acyclic, for all
$y$.
\end{proposition}

\begin{proof}
Let $c$ be rationalizable by the additive privacy preferences characterized by
$(u,v)$.  For each $x,z \in X$ such that $x\mathbin{R^y} z$, then
there is some $A\in \A$ such that $z=c(A)$ and $x\in A$,  so
\[ u(z) - \sum_{t\in A} v(z,t) > u(x) - \sum_{t\in A} v(x,t).\]
Similarly, $x= c(A\cup\{y\})$ and $z \in A \cup \{y\}$, so
\[ u(z) - \sum_{t\in A} v(z,t) - v(z,y) > u(x) - \sum_{t\in A} v(x,t) - v(x,y).\]
For both inequalities to be true simultaneously, we need $v(z,y) > v(x,y)$.  Thus,
\begin{equation}\label{eq:uvcycle} x\mathbin{R^y} z \Longrightarrow
  v(z,y) > v(x,y). \end{equation}

Now assume there exists a cycle in binary relation $R^y$: $a_1
\mathbin{R^y} a_2 \mathbin{R^y} \cdots \mathbin{R^y} a_k \mathbin{R^y}
a_1$.  Then by Equation~\eqref{eq:uvcycle}, it must be that $v(a_1, y)
> v(a_2,y) > \cdots > v(a_k, y) > v(a_1,y)$.  In particular, $v(a_1,
y) > v(a_1,y)$ which is a contradiction.  Then for choices to be
rationalized, acyclicity of $R^y$ for all $y \in X$ is a necessary
condition.
\end{proof}

Of course, one would like to develop a test for rationalizability that
is both necessary and sufficient.  We do this next.  Unfortunately,
the test we develop takes super-exponential time to even write down.
This suggests that acyclicity of $R^y$, despite being only a necessary
condition, is likely a more practical condition to use when testing
for rationalizability.

To describe the test for rationalizability, first observe that when an
object $x$ is chosen from a set, the observer infers that $x$ (with
its associated privacy) is preferred to $y$ (and its associated
privacy), for all $y \in A \setminus \{x\}$.  Since we have assumed
these preferences to have a specific functional form as in
Definition~\ref{def:addpref}, the observer can also infer the
corresponding inequality in terms of functions $u$ and $v$.  We
initialize a large matrix to record the inequalities that are inferred
from choice behavior, and ask if there exist values of $u(x)$ and
$v(x,x')$ for all $x, x' \in X$ for which all inferred inequalities
hold.  If so, these values of $u$ and $v$ form additive privacy preferences that
rationalize choices.  If not, then no such preferences exist and the
observed choice behavior is not rationalizable.

\begin{proposition} \label{prop:necsuff}
A choice problem $(X, \A, c)$ is rationalizable if and only if there exists functions $u:X \rightarrow \mathbb{R}^+$ and $v:X \times X \rightarrow \mathbb{R}^+$ satisfying the matrix inequality given by equation~\eqref{eq:matrix}.
\end{proposition}
\begin{proof}
For ease of notation, index the elements of $X = \{x_1, \ldots, x_n\}$.  Then for each $A \in \A$, the agent chooses some $x_i = c(A) \in A$.  By the definition of additive preferences, every $x_j \in A$ for $j \neq i$ was \emph{not} chosen because
\[u(x_i) - \sum_{z \in A \setminus \{x_i\}} v(x_i,z) > u(x_j) - \sum_{z \in A \setminus \{x_j\}} v(x_j,z) \]
Rearranging terms gives,
\begin{equation}\label{eq:obs} u(x_i) - u(x_j) + \sum_{z \in A \setminus \{x_j\}} v(x_j,z) - \sum_{z \in A \setminus \{x_i\}} v(x_i,z) > 0 \end{equation}
We can instantiate a matrix $T$ to record the inequalities implied by all observed choices.   This matrix $T$ will have $n^2$ columns, where the first $n$ columns correspond to elements $x_1, \ldots, x_n \in X$, and the remaining $n^2 - n$ columns correspond to ordered pairs $(x_i, x_j)$ of elements in $X$, for $i \neq j$.\footnote{There are only $n^2 - n$ columns because we do not compare elements to themselves}  $T$ will contain a row for each triple $(A, x_i, x_j)$, where $A \in \A$, and $x_i, x_j \in A$.  If the agent is observed to choose $x_i = c(A)$, then Equation~\eqref{eq:obs} must be true for each $x_j \in A$ for $j \neq i$.  To encode this inequality, then for each such $x_j$ that was not chosen from $A$, we fill in the row corresponding to $(A, x_i, x_j)$ as follows: enter $+1$ in the $i^{th}$ column, $-1$ in the $j^{th}$ column, $+1$ in columns corresponding to pairs $(x_j, z)$ where $z \in A$, $-1$ in columns corresponding to pairs $(x_i, z)$ where $z \in A$, and zeros elsewhere.

To complete the encoding, we also require a vector $\vec{u}$, which contains the values of $u(\cdot)$ and $v(\cdot, \cdot)$ evaluated on all elements of $X$.  The first $n$ entries of $\vec{u}$ will contain the values of $u(x_1), \ldots, u(x_n)$, and the remaining $n^2 -n$ entries will contain the values $v(x_i, x_j)$ for $i\neq j$, in the same order in which the pairs appear in the columns of $T$.  Then multiplying $T \vec{u}$, each row of the product would equal
\[ u(x_i) - u(x_j) + \sum_{z \in A \setminus \{x_j\}} v(x_j,z) - \sum_{z \in A \setminus \{x_i\}} v(x_i,z) \]
for some set $A \in \A$, observed choice $x_i = c(A)$, and unchosen element $x_j \in A$.  Finally, we need the constraint that each row of $T \vec{u}$ is greater than zero, as required in Equation~\ref{eq:obs}.  That is,
\begin{equation}\label{eq:matrix} T \vec{u} > \vec{0} \end{equation}
If such a vector $\vec{u}$ exists, then there exist functions $u:X \rightarrow \mathbb{R}^+$ and $v:X \times X \rightarrow \mathbb{R}^+$ such that additive privacy preferences are optimized by the observed choices, and the thus observed choices are rationalizable.
\end{proof}

\section{Higher order privacy preferences}
\label{s.higherorder}

The results we have discussed so far are predicated on the notion that
the agent thinks that the observer is naive. We shall now relax the
assumption of naivete. We are going to
allow the agent to believe that the observer thinks that she is privacy
aware.

Going back to Alice, who is choosing among political
blogs, suppose that she reasons as follows. Alice may realize that her
observed choices violate the strong axiom of revealed preference and
therefore cannot correspond to the choices of a rational agent. This
could tip off the observer to the fact that she is privacy aware. We
have seen that privacy awareness is a plausible explanation for
violations of the revealed preference axioms. So Alice could now be
concerned about the observer's inference about her preferences over
objects and over revealed preference. Perhaps she thinks that the
observer will infer that she is avoiding blog $a$ because of what it
reveals about her, and that fact itself is something she does not wish
be  known. After all, if Alice has a preference for privacy, perhaps
she has something to hide.

More generally,  an agent may be
concerned not only about what her behavior reveals about her
preferences over $X$, but also about what her behavior reveals of her
preferences for privacy.  She may then make choices to minimize
inferences the observer is able to make about her preferences for
privacy, as well as her preferences over $X$.

To provide a model that incorporates such issues,
we define a hierarchy of higher order preferences, called {\em
  level-$k$ preferences}, where a level-$k$ consumer is aware that the
observer may make inferences about her level-$(k-1)$ privacy
preferences, and has preferences over the information the observer can
infer. In our construction, level-$0$ corresponds to the classical
privacy-oblivious setting, and the setting we have considered to this
point is that of a level-$1$ consumer (Sections~\ref{s.modeling} and
\ref{s.mainresults}).

The meaning of such levels should be clear. If Alice is concerned
about facing an observer who makes level $k$ inferences, then her
behavior will be dictated by the level $k+1$ model. To emphasize a
point we have made repeatedly, {\em the real observer may be as
sophisticated as one wants,} but Alice thinks that the observer thinks
that Alice thinks that the observer thinks
that Alice thinks \ldots that the observer makes inferences based on
revealed preferences.

\subsection{Level-$k$ privacy preferences}
\label{ss.levelk}

To formally define a ``cognitive hierarchy'' for privacy-aware
consumers we use the following sequence of sets, $\Y^k$ for $k \geq
0$.  $\Y^0 = X$,
$\Y^1 = X\times \B(\Y^0),$ and let
$\Y^k = X\times \B(\Y^{k-1})$.
A level-$k$ privacy preference can then be defined as a binary relation $\succeq^k$ over  $\Y^k = X\times \B(\Y^{k-1})$. That is, $\succeq^k$ describes preferences over pairs of objects $x \in X$ and the set of level-$(k-1)$ preferences that are revealed from the choice of $x$.

Given the results in Section \ref{s.mainresults}, our focus is on monotone, separable privacy preferences, and so we can extend the notion of monotonicity discussed in Section \ref{s.mainresults} to level-$k$ privacy preferences as follows.

\begin{definition}
A \textbf{monotone level-$k$ privacy preference} is a binary relation $\succeq^k$ over $\Y^k = X\times \B(\Y^{k-1})$ such that
\begin{enumerate}
\item $\succeq^k$ is a linear order, and
\item $B\subsetneq B'$ implies that $(x,B)\succ (x,B')$, for all $B,B'\in\B(\Y^{k-1})$.
\end{enumerate} For this definition to hold for level-0, we define $\Y^{-1}$ to be the empty set.
\end{definition}

Similarly, we extend the notion of separability to level-$k$ privacy preferences as follows.

\begin{definition}
A \textbf{separable level-$k$ privacy preference} is a binary relation $\succeq^k$ over $\Y^k = X\times \B(\Y^{k-1})$ such that it is monotone and additionally satisfies for any $B\in\B(\Y^{k-1})$,
\[(x, B) \succeq^k (y,B) \Longrightarrow (x, B^{\prime}) \succeq^k (y,B^{\prime}) \; \forall B^{\prime}\in\B(\Y^{k-1})\]
\end{definition}

Given the notion of level-$k$ privacy preferences, we need to characterize how an observer will make inferences from observed choices.  Naturally, the exact information inferred will depend on the level which the observer believes the privacy preferences to be.
For example, if the observer believes the consumer to have level-$0$ preferences, the information inferred by the observer is the set
\[A_x =  \{(x,y) : y\in A\setminus \{x\}\},\]
which is a binary relation over $X$. So $A_x\in \B(\Y^0)$.
However, if the observer believes the consumer to have level-$1$ preferences, the information inferred by the observer is the set
\[\{((x,A_x), (y,A_y)) : y\in A\setminus \{x\}\}\in \B(\Y^1).\]

More generally, to describe the observer's inferences under the the belief that the consumer is level-$k$,  we introduce the following notation. Consider the functions $T^k:\A\times X\rightarrow \B(\Y^k)$, for $k\geq 0$. Let
\begin{align*}
T^0(A,x) & = \{(x,y) : y\in A\setminus \{x\}\} \in \B(\Y^0) \\
T^1(A,x) & = \{\left( (x,T^0(A,x)), (y,T^0(A,y)) \right) : y\in A\setminus \{x\}\} \in \B(\Y^1) \\
\vdots \mbox{ }& \mbox{ }\;\;\;\,\, \vdots \\
T^k(A,x) & = \{\left( (x,T^{k-1}(A,x)), (y,T^{k-1}(A,y)) \right) : y\in A\setminus \{x\}\}\in \B(\Y^k).
\end{align*}

In words, $T^k(A,x)$ are the level-$k$ preferences (over alternatives
in $A$ and set of level-$(k-1)$ preferences that will be inferred from
each choice) that would cause the agent to choose $x$ from the set
$A$.  Then generally, a level-$k$ agent making choice $x=c(A)$ must
have $T^k(A,x)$ as a subset of her level-$k$ preferences.

\subsubsection*{Example: Level-$2$ privacy preferences}

In order to illustrate the cognitive hierarchy more concretely it is
useful to describe the case of level-$2$ privacy preferences in
detail.
Recall that the level-$0$ privacy preferences are the
classical setting of privacy-oblivious consumers and level-$1$ privacy
preferences are the case we study in Sections~\ref{s.modeling}
and~\ref{s.mainresults}.  As we shall see, there is a sense in which
level-$2$ is all that is needed.

Continuing with the story about Alice, we remarked how she could come
to question her level-1 behavior because she should realize that there
is something suspicious about her choices violating the revealed
preference axioms. As the result of such a realization, she might
entertain level-$2$ behavior. She might think that the observer thinks
that she is level-$1$. Now, there is no reason for her to go any
further because, in contrast with level-$1$, {\em nothing could give
  her away.}

While her violations of the revealed preference axioms
indicate that she cannot be level-$0$, given our
Proposition~\ref{prop:anyc},
nothing about her behavior could contradict that she is level-$1$.  She has no reason to think that reasoning beyond level-$2$ will afford her more privacy - we have already seen that nothing in her behavior that could prove to the observer that she is not level-$1$.

More concretely, suppose that $x$ is chosen from set $A$. The observer, who thinks the consumer is at level-$1$, infers the level-$1$ preferences
\[ (x,A_x) \succ (z,A_z) \; \; \forall \; z\in A\setminus\{x\},\]
or, more specifically, that her level-$1$ privacy preferences correspond to the binary relation, \begin{equation}\label{eq:level2}
\bigcup \left\{ [(x,A_x), (z,A_z)] : z\in A\setminus\{x\} \right\}.
\end{equation}

Now the agent who believes that
the observer will make such an inference, will only choose $x$
when this choice {\em together with inferences revealed by the choice}
is better than the choice of another alternative in $A$ with its
accompanying inferences. That is, she will choose $x$ over $y$ in $A$
whenever the choices of $x$ {\em together} with the release of the
information in Equation~\eqref{eq:level2} is preferred to the choice
of $y$ together with the information,
\[ \bigcup \{ [(y,A_y), (z,A_z)] : z\in A\setminus\{y\} \}. \]

That is, if a level $2$ agent chooses $x$ from set $A$, she knows that observer will make inferences according to Equation~\eqref{eq:level2}.  Then her choice of $x$ must maximize her preferences over outcomes \emph{and} these known inferences that will be made.  Specifically, she will choose $x$ if her level-$2$ preferences are, for all available $y \in A$,
\begin{equation}\label{eq:level1rev} (x, \cup \{ [(x,A_x), (z,A_z)] : z\in A\setminus\{x\} \}) \succ (y, \cup \{ [(y,A_y), (z,A_z)] : z\in A\setminus\{z\} \}) \end{equation}

Using the notation defined earlier in this section, we can re-write Equation~\eqref{eq:level1rev} as a binary relation,
\[[(x, T^1(A,x)), (y, T^1(A,y))]. \]

Since the same can be said for {\em every} available alternative $y\in A$ that was not chosen, the following must be a part of the agent's level-$2$ preferences
\[ T^2(A,x) = \{\left[ (x,T^1(A,x)), (y,T^1(A,y)) \right] : y\in A\setminus \{x\}\} \]
Note, however, that the observer does not get to infer $T^2(A,x)$.  He believes the agent to have level-$1$ preferences, and upon seeing $x=c(A)$, he infers $T^1(A,x)$.  This is why the agent chooses $x \in A$ to optimize her preferences over $X$ \emph{and} sets of the form $T^1(A,\cdot )$.

\subsection{The rationalizability of level-$k$ preferences}

Given the notion of a privacy-aware cognitive hierarchy formalized by
level-$k$ privacy preferences, we are now ready to move on to the task
of understanding the empirical implications of higher order reasoning
by privacy-aware consumers. To do this, we must first adapt the notion
of rationalizability to level-$k$ reasoning.  For this, the natural
generalization of Definition \ref{def:rational} to higher order
reasoning is as follows.  This definition reduces to Definition
\ref{def:rational} when level-$1$ is considered, and to the classical
definition of rationalizable in the privacy-oblivious case when
level-$0$ is considered.

\begin{definition}
A choice $(X,\A,c)$ is \textbf{level-$k$ rationalizable} if there is a level-$k$ privacy preference $\succeq^k\in \B(\Y^k)$ such that for all $A\in\A$, $T^k(A,c(A))\subseteq \succeq^k$.
\end{definition}

Given this definition, we can now ask the same two questions we
considered in Section \ref{s.mainresults} about level-$k$ privacy
preferences:  ``What are the empirical implications of level-$k$
privacy preferences?'' and ``When can the observer learn the
underlying choice preferences of consumers?''  Our main result is
the following theorem, which answer these questions.

\begin{theorem} \label{thm:main}
Let $(X,\A,c)$ be a choice problem. Let $k>0$ and $\succeq$ be any linear order over $X$. Then there is a monotone, separable level-$k$ privacy preference $\succeq^*$ that level-$k$ rationalizes $(X,\A,c)$ and such that:
\[ x\succeq y \text{ iff } (x, B) \succeq^* (y,B) \text{ for all } B\in\B(\Y^{k-1}). \]
\end{theorem}

\begin{proof}

Let $T^{k-1}: \A\times X \rightarrow \B(\Y^{k-1})$ be as defined in
Section~\ref{ss.levelk}.  For shorthand, write $\Y$ for $\Y^{k-1}$ and $T$ for $T^{k-1}$. Then $T$ describes the set of level-$(k-1)$ preferences inferred by the observer as a result of the agent's choice behavior.  That is, when the agent chooses $x=c(A)$, the observer will infer all preferences in the set $T(A,x)$.
Note that $T$ is one-to-one and satisfies the following property: for all $A\in \A$ and all $x,x'\in A$,
\begin{equation}\label{eq:propT}|T(A,x)| = |T(A,x')|. \end{equation}
Property~\eqref{eq:propT} follows because the number of pairs$((x,T(A,x)), (y,T(A,y)))$ $y\in A\setminus \{x\}\}$ is the same for any $x\in A$.

We now construct a binary relation $E$ over $X \times \B(\Y)$.  As in the proof of Proposition~\ref{prop:anyc}, it will be useful to consider $E$ as the edges of a directed graph $G=(V,E)$, where $V = X\times \B(\Y)$.  We create edges in $E$ according to the desiderata of our privacy-aware preferences: monotone, separable, and rationalizing choice behavior.  Define $E$ as follows: $(x,B) \mb E (x',B')$ if either (1) $x=x'$ and $B\subsetneq B'$, (2) $(x,B) \mb E (x^{\prime},B)$, where $x\succ x^{\prime}$ according to linear order $\succeq$, or (3) $x\neq x'$ and there is $A\in \A$ with $x=c(A)$, $x'\in A$ and $B=T(A,x)$ while $B'=T(A,x')$.  We will call these edges respectively ``monotone,'' ``separable,'' and ``rationalizing,'' as a reference to the property they are meant to impose.  By Lemma~\ref{lem:spil}, we are done if we show that $E$ is acyclic.

Assume towards a contradiction that there is a cycle in this graph. Then there exists a sequence $j=1, \ldots, K$ such that \[(x^1, B^1) \mb E (x^2, B^2) \mb E \cdots \mb E (x^K, B^K) \text{ and }(x^K, B^K) \mb E (x^1, B^1).\]

For any monotone edge $(x^i,B^i) \mb E (x^{i+1},B^{i+1})$, it must be the case that $|B^i| < |B^{i+1}|$ since $B^i \subset B^{i+1}$.  If this were a separable edge, then $B^i = B^{i+1}$, so $|B^i| = |B^{i+1}|$.  Similarly, for any rationalizing edge, $|B^i| = |T(A,x)| = |T(A,x')| = |B^{i+1}|$.  Thus as we traverse any path in this graph, the size of the second component is non-increasing along all edges, and strictly decreasing along monotone edges.  This implies that there can be no cycles containing monotone edges, and our assumed cycle must consist entirely of rationalizing and separable edges.

If there are two sequential rationalizing edges in this cycle, then there exists a $j$ and some $A \in \A$ such that \[(x^j, T(A,x^j)) \mb E (x^{j+1}, T(A, x^{j+1})) \mb E (x^{j+2}, T(A, x^{j+2})),\] where $x^j, x^{j+1}, x^{j+2} \in A$.  From the first edge, $x^j = c(A)$ in some observation, and from the second edge, $x^{j+1} = c(A)$ in another observation.  If $x^j \neq x^{j+1}$ then $c(A) = x^j \neq x^{j+1} = c(A)$ which contradicts the uniqueness of choice imposed by the linear ordering.  If $x^j = x^{j+1}$, then $(x^j, T(A,x^j)) = (x^{j+1}, T(A, x^{j+1}))$, which implies that an element of $X \times \B(\Y)$ is strictly preferred to itself, which is a contradiction. Thus no cycle can contain two sequential rationalizing edges.

If there are two sequential separable edges in our cycle, then there exists a $j$ such that
\[(x^j, B^j) \mb E (x^{j+1}, B^{j+1}) \mb E (x^{j+2}, B^{j+2}),\] where $B^j = B^{j+1} = B^{j+2}$ and $x^j\succ x^{j+1}\succ x^{j+2}$. By transitivity, $x^j\succ x^{j+2}$, so there must also be a separable edge in the graph $(x^j, B^j) \mb E (x^{j+2}, B^{j+2})$.  If the cycle we have selected contains two sequential separable edges, then there must exist another cycle that is identical to the original cycle, except with the sequential separable edges replaced by the single separable edge.  Thus we can assume without loss of generality that the cycle we have selected does not contain two sequential separable edges.\footnote{The only time this is with some loss of generality is when there is a cycle containing only separable edges.  By assumption, $\succeq$ is a linear order over $X$ and must be acyclic, so this is not possible.}

Given the previous two observations, we can assume without loss of generality that the cycle $j=1, \ldots, K$ contains alternating rationalizing and separable edges.  This includes the endpoints $j=K$ and $j=1$ since they too are connected by an edge.

If there is a path in the graph containing sequential rationalizing, separable, and rationalizing edges, then there exists $A, A^{\prime} \in \A$ and a $j$ such that \[ (x^j, T(A, x^j)) \mb E (x^{j+1}, T(A, x^{j+1})) \mb E (x^{j+2}, T(A^{\prime}, x^{j+2})) \mb E (x^{j+3}, T(A^{\prime}, x^{j+3})), \] where $x^j, x^{j+1} \in A$, $x^{j+2}, x^{j+3} \in A^{\prime}$, and choices $x^j = c(A)$ and $x^{j+2} = c(A^{\prime})$ are observed.  Since the middle edge is separable, it must be that $T(A,x^{j+1}) = T(A^{\prime}, x^{j+2})$, so $x^{j+1} = x^{j+2}$ and $A = A^{\prime}$.  However, this means that $(x^{j+1}, T(A, x^{j+1}))$ is strictly preferred to itself, which is a contradiction, so no such path in the graph can exist.

Since edges must alternate between rationalizing and separable, this leaves the only possible cycles to be of length two, containing one rationalizing edge and one separable edge.  However, if such a cycle existed, then traversing the cycle twice would yield a path containing sequential rationalizing, separable, and rationalizing edges, which has been shown to not exist in this graph.

Thus we can conclude that $E$ must be acyclic, which completes the proof.
\end{proof}

\section{Related work}

The growing attention to privacy concerns has led to a growing literature studying privacy, see \cite{heffetz2013privacy} for a survey.  Within this literature, an important question is how to model the preferences or utilities of privacy-aware agents in a way that describes their behavior in strategic settings.

One approach toward this goal, exemplified by \cite{GR11}, \cite{NOS12}, \cite{CCKMV13}, \cite{Xiao13}, and \cite{NVX14}, is to use differential privacy in mechanism design as a way to quantify the privacy loss of an agent from participating the mechanism.  Within this literature, each of \cite{GR11}, \cite{NOS12}, and \cite{Xiao13} assume that the utility of a privacy-aware agent is her gain from the outcome of the interaction minus her loss from privacy leakage.  Note that this is a stronger condition than separability, as defined in Section~\ref{s.separable}, and a weaker condition than additivity, as defined in Section~\ref{s.additive}. In contrast,   \cite{CCKMV13} and \cite{NVX14} make the same separability assumption as used in this paper, but \cite{CCKMV13} allows for non-monotone privacy preferences and \cite{NVX14} uses a relaxed version of monotonicity.

Perhaps the model closest to ours is \cite{Gra13}, which also considers privacy-aware agents with preferences over outcome-privacy pairs.  However, the technical quantification of privacy is different in the two models, as \cite{Gra13} considers multiple agents engaging in a single interaction instead of multiple choices by a single agent as in the current paper.  In addition, the nature of the results in \cite{Gra13} are different from ours: it studies implementation from a mechanism design perspective, while we study the testable implications of privacy-aware preferences through the lens of revealed preference analysis.

\section{Concluding remarks}
\label{s.conclusion}

We conclude by describing what our results mean for Alice's story, and
for future research on privacy.

Alice makes choices knowing that she is being observed. She thinks
that the observer uses revealed preference theory to infer her
preferences. She might think that the observer is not sophisticated,
and uses revealed preferences naively to infer her preferences over
objects. Alternatively, she might think that the observer is sophisticated, and
knows that she has preferences for privacy; in this case, the
observer tries to infer (again using revealed preferences) Alice's
preferences for objects and privacy.

The story of Alice, however, is more ``The Matrix'' than ``in
Wonderland.'' Alice believes that she is one step ahead of the
observer, and makes choices taking into account what he learns about
her from her choices. In reality, however, the observer is us: the
readers and writers of this paper.

{\em We} are trying to understand Alice's behavior, and to infer what
her preferences over objects might be. The main result of our work is
that such a task is hopeless. Any behavior by Alice is consistent with
any preferences over objects one might conjecture that she has (and
this is true for any degree of sophistication that Alice may have in
her model of what the observer infers from her).
Other observers on the internet, such as Google or the NSA, would have
to reach the same conclusion.

One way out is to impose additional structure on Alice's
preferences. The main result uses separability and monotonicity, which are
strong assumptions in many other environments, but that is not
enough. We have suggested additive privacy preferences
(Section~\ref{s.additive}) as a potentially useful model to
follow. Additive preferences do impose observable restrictions on
choice, and its parameters could be learned or estimated from choice
data. Privacy researchers looking to model a utility for privacy
should consider the additive model as a promising candidate.

\bibliography{privacy}

\begin{thebibliography}{10}

\bibitem{CCKMV13}
Yiling Chen, Stephen Chong, Ian~A. Kash, Tal Moran, and Salil Vadhan.
\newblock Truthful mechanisms for agents that value privacy.
\newblock In {\em Proceedings of the 14th ACM Conference on Electronic
  Commerce}, EC '13, pages 215--232, 2013.

\bibitem{GR11}
Arpita Ghosh and Aaron Roth.
\newblock Selling privacy at auction.
\newblock In {\em Proceedings of the 12th ACM Conference on Electronic
  Commerce}, EC '11, pages 199--208, 2011.

\bibitem{Gra13}
Ronen Gradwohl.
\newblock Privacy in implementation.
\newblock In {\em CMS-EMS Discussion Paper 1561}, 2013.

\bibitem{heffetz2013privacy}
Ori Heffetz and Katrina Ligett.
\newblock Privacy and data-based research.
\newblock Technical report, National Bureau of Economic Research, 2013.

\bibitem{mas1995}
Andreu Mas-Colell, Michael~D. Whinston, and Jerry~R Green.
\newblock {\em Microeconomic theory}, volume~1.
\newblock Oxford university press New York, 1995.

\bibitem{NOS12}
Kobbi Nissim, Claudio Orlandi, and Rann Smorodinsky.
\newblock Privacy-aware mechanism design.
\newblock In {\em Proceedings of the 13th ACM Conference on Electronic
  Commerce}, EC '12, pages 774--789, 2012.

\bibitem{NVX14}
Kobbi Nissim, Salil Vadhan, and David Xiao.
\newblock Is privacy compatible with truthfulness?
\newblock In {\em Proceedings of the 4th Innovations in Theoretical Computer
  Science}, ITCS '14, 2014.
\newblock To appear.

\bibitem{pewsurvey}
Lee Rainie, Sara Kiesler, Ruogu Kang, and Mary Madden.
\newblock Anonymity, privacy and security online.
\newblock Technical report, Pew Research Center, 2013.

\bibitem{rubinstein2012lecture}
Ariel Rubinstein.
\newblock {\em Lecture notes in microeconomic theory: the economic agent}.
\newblock Princeton University Press, 2012.

\bibitem{simonson1992choice2}
Itamar Simonson and Amos Tversky.
\newblock Choice in context: tradeoff contrast and extremeness aversion.
\newblock {\em Journal of marketing research}, 1992.

\bibitem{varia82}
Hal~R. Varian.
\newblock The nonparametric approach to demand analysis.
\newblock {\em Econometrica}, 50(4):945--974, Jul 1982.

\bibitem{varian2006revealed}
Hal~R. Varian.
\newblock {Revealed preference}.
\newblock {\em Samuelsonian economics and the twenty-first century}, pages
  99--115, 2006.

\bibitem{Xiao13}
David Xiao.
\newblock Is privacy compatible with truthfulness?
\newblock In {\em Proceedings of the 4th Innovations in Theoretical Computer
  Science}, ITCS '13, pages 67--86, 2013.

\end{thebibliography}
\bibliographystyle{plain}

\end{document}